\documentclass[aps,prx,twocolumn,10pt,superscriptaddress,floatfix]{revtex4-2}
\usepackage[utf8]{inputenc}
\usepackage[english]{babel}
\usepackage[T1]{fontenc}
\usepackage{amsmath, amssymb, amsfonts, amsthm}
\usepackage{bbm}
\usepackage{hyperref}
\usepackage{xparse}
\makeatletter
\newcounter{circuit}
\newenvironment{circuit}[1][tbph] 
  {\let\c@figure\c@circuit           
   \let\fnum@figure\fnum@circuit     
   \begin{figure}[#1]}
  {\end{figure}}
\newenvironment{circuit*}[1][tbph]
  {\let\c@figure\c@circuit
   \let\fnum@figure\fnum@circuit
   \begin{figure*}[#1]}
  {\end{figure*}}
\newcommand{\fnum@circuit}{Circuit~\thecircuit:}
\makeatother
\usepackage[capitalize,nameinlink]{cleveref}
\Crefname{figure}{Figure.}{Figs.}
\Crefname{tabular}{Table.}{Tabs.}
\Crefname{section}{\S}{\S}
\Crefname{theorem}{Theorem.}{Thms.}
\Crefname{lemma}{Lemma.}{Lems.}
\Crefname{corollary}{Corollary.}{Cors.}
\Crefname{algorithm}{Algorithm.}{Algs.}
\Crefname{example}{Example.}{Exs.}
\Crefname{definition}{Definition.}{Defs.}
\Crefname{proof}{Proof.}{Proofs.}
\crefname{circuit}{Circuit.}{Circuits.}

\usepackage{xspace}
\usepackage{booktabs}
\usepackage{multirow}
\usepackage{array}
\usepackage{nicefrac}

\usepackage{algorithm}
\usepackage[noend]{algpseudocode}

\usepackage{tikz}
\usetikzlibrary{positioning, arrows.meta}
\usetikzlibrary{quantikz2}
\usepackage{braket}

\newcommand{\Ry}{R_y\xspace}
\newcommand{\Rz}{R_z\xspace}
\newcommand{\Uprod}{U_{\text{prod}}\xspace}
\newcommand{\Usum}{U_{\text{sum}}\xspace}
\newcommand{\Uf}{U_{\text{flip}}\xspace}

\newcommand{\qubit}{\ensuremath{q}\xspace}

\newcommand{\cx}{\ensuremath{CX}\xspace}

\newcommand{\rza}[1]{\ensuremath{R_z^{{#1}}}\xspace}
\newcommand{\rya}[1]{\ensuremath{R_y^{{#1}}}\xspace}

\newcommand{\ours}{\textbf{QAWA}\xspace}

\newtheorem{theorem}{Theorem}
\newtheorem{lemma}{Lemma}
\newtheorem{dfn}{Definition}
\newtheorem{prop}{Proposition}
\newtheorem{claim}{Claim}
\newtheorem{corol}{Corollary}

\renewcommand{\paragraph}[1]{\textbf{{\emph {#1}}.~~~}}

\begin{document}

\title{Quantum Approximate Walk Algorithm}

\author{Ziqing Guo}
\affiliation{Texas Tech University, Department of Computer Science, TX, USA}
\affiliation{Lawrence Berkeley National Laboratory, NERSC, CA, USA}
\email{Corresponding authors: ziqguo@ttu.edu, ziwen.pan@ttu.edu}

\author{Jan Balewski}
\affiliation{Lawrence Berkeley National Laboratory, NERSC, CA, USA}

\author{Wenshuo Hu}
\affiliation{Texas Tech University, Department of Chemical Engineering, TX, USA}

\author{Alex Khan}
\affiliation{University of Maryland, National Quantum Laboratory, MD, USA}

\author{Ziwen Pan}
\affiliation{Texas Tech University, Department of Computer Science, TX, USA}
\email{Corresponding authors: ziqguo@ttu.edu, ziwen.pan@ttu.edu}

\begin{abstract}
The encoding of classical to quantum data mapping through trigonometric functions within arithmetic-based quantum computation algorithms leads to the exploitation of multivariate distributions.
The studied variational quantum gate learning mechanism, which relies on agnostic gradient optimization, does not offer algorithmic guarantees for the correlation of results beyond the measured bitstring outputs. Consequently, existing methodologies are inapplicable to this problem.
In this study, we present a classical data-traceable quantum oracle characterized by a circuit depth that increases linearly with the number of qubits. This configuration facilitates the learning of approximate result  patterns through a shallow quantum circuit (SQC) layout.
Moreover, our approach demonstrates that the classical preprocessing of mid-quantum measurement data enhances the interpretability of quantum approximate optimization algorithm (QAOA) outputs without requiring full quantum state tomography.
By establishing an inferable mapping between the classical input and quantum circuit outcomes, we obtained experimental results on the state-of-the-art IBM Pittsburgh hardware, which yielded polynomial-time verification of the solution quality.
This hybrid framework bridges the gap between near-term quantum capabilities and practical optimization requirements, offering a pathway toward reliable quantum-classical algorithms for industrial applications.
\end{abstract}

\maketitle

\section{Introduction}
The quantum information theory underlies contemporary noisy intermediate-scale quantum (NISQ) computation and its inherent parallelism \cite{preskill2018quantum}. Key practical algorithms include the Quantum Approximate Optimization Algorithm (QAOA) for combinatorial optimization \cite{qaoa}, Harrow–Hassidim–Lloyd (HHL) algorithm for linear systems \cite{PhysRevLett.103.150502}, Shor’s factoring algorithm \cite{shor1994algorithms}, and Grover’s unstructured search procedure \cite{grover1996fast}. Their power is derived from exploring exponentially large solution spaces in time, independent of the problem size, while requiring only polynomially many qubits \cite{smith2019simulating}. Quantum approximation is particularly relevant to large‐scale tasks such as portfolio optimization \cite{markowitz1952portfolio}; within a high-dimensional parameter landscape it attains near-optimal results with a finite set of bounded parameterized quantum gates of limited circuit depth, thereby necessitating a specific number of qubits \cite{martyn2021grand}.

This approximation algorithm falls within the domain of quantum circuit learning (QCL). From a practical standpoint, emerging learning models such as variational quantum algorithms (VQAs) \cite{moll2018quantum}, quantum walks (QWs) \cite{childs2009universal,santha2008quantum}, and quantum Boltzmann machines (QBMs) \cite{amin2018quantum} address general optimization tasks \cite{finance_all,cornuejols2018optimization,rebentrost2018quantum,herman2023quantum} by minimizing the cost objective using gradients obtained from quantum measurements. Current QCL frameworks approximate high-dimensional functions by encoding classical data into quantum kernels \cite{balewski2025ehands}; however, their efficacy diminishes as the complexity of the unitary operations and sampling iterations increases. In particular, performance depends on internal correlations and sampling biases extracted from the measured outputs and circuit configurations, which can mask the true optimum. 
Nevertheless, quantum circuit solvers can be confined to the SQC regime \cite{bravyi2018quantum}, which remains challenging to simulate classically, but can be solved on two-dimensional grid quantum computers \cite{briegel2009measurement}.
Although qubit resources are restricted to discrete quantum computation, pioneering studies \cite{balewski2024quantum,balewski2025compilation,amankwah2022quantum} have demonstrated that classical‐to‐quantum encodings exploit quantum parallelism attainable on near‐term, error‐prone quantum processing units (QPUs) \cite{sharma2025evaluation,PRXQuantum.3.010347,kim2025error}.

\begin{figure*}[t]
    \centering
    \includegraphics[width=0.99\linewidth]{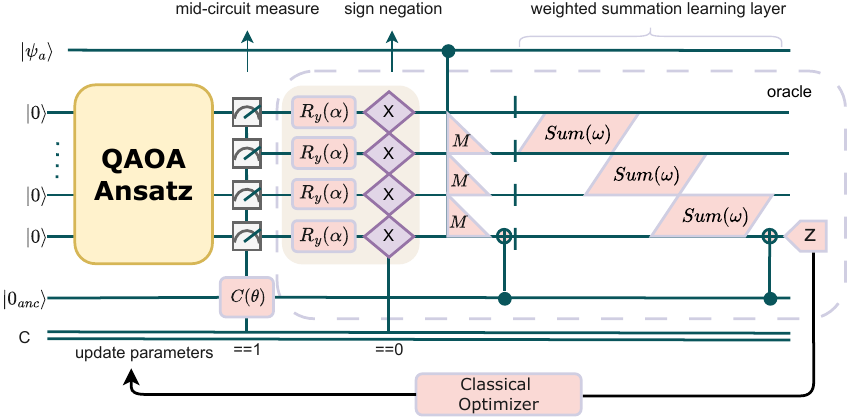}
    \caption{Quantum Approximate Walk Algorithm (\ours) architecture: mid-circuit measurements applied to the QAOA ansatz yields outcomes that feed an $\mathrm{R}_y$ encoding layer with sign negation; the unitary block $C(\theta)$ generates a walkable state that, via one ancilla qubit, controls the weighted-summation learning layer conditioned on the measurement results. The top register $\ket{\psi_a}$ functions as an activation multiplier that globally modulates the learning procedure, and the training loop employs a classical optimizer.}
    \label{fig:main}
\end{figure*}

To improve approximation accuracy, several methodological advances have been explored, including heuristic parameter initialization for the QAOA ansatz \cite{ sanders2020compilation, nannicini2019performance, he2023alignment, guo2025direct}, fixed-ramp scaling of variational parameters \cite{ montanez2025toward}, generative quantum adversarial networks (QGANs) \cite{ zhu2022generative, niu2022entangling, zoufal2019quantum}, gradient–measurement trade-off strategies \cite{ chinzei2025trade}, and deterministic quantum walks \cite{ peng2024deterministic}. These techniques yield favorable scaling and implicitly leverage the approximated QCL structure shaped by underlying physical and classical data; however, definitive links between measured outcomes and their sampling distributions remain elusive. For concreteness, \cite{ shaydulin2024evidence} demonstrates that problem energy spectra modulate bitstring probabilities; similarly, \cite{ herrman2021impact} connects graph symmetry and density to algorithmic success rates.

What can be demonstrated to derive the outcome correlation structure from a quantum approximation algorithm? Alternatively, how can the approximation be evaluated? Most importantly, why do approximation algorithms yield specific correlation patterns? Although the success ratio directly reflects the algorithm performance, a more detailed examination of each bitstring dependence is warranted. A native quantum arithmetic protocol (Ehands) \cite{balewski2025ehands} incorporates block data encoding \cite{camps2022fable} for multiplication, weighted summation, sign negation, and parity flip, which are independent of the fully entangled VQA learning ansatz \cite{cerezo2021variational}. This particular design enables the estimation of observables based on real classical input; thus, by integrating the results derived from the quantum approximation ansatz, tunable classical-to-quantum encoded data can manifest a result correlation. Consequently, building upon our previous work \cite{guo2025vectorized} on the vectorized quantum learning scheme, we introduce a novel quantum approximate walk algorithm (\ours) that provides an asymptotic solution to substantiate the quantum approximation results, even when the quantum circuit is composed and transpiled on QPUs with limited capabilities. By eliminating the conventional quantum state transition scheme (known as state hopping or walking), our approach effectively captures each probabilistically distributed outcome derived from classical definitive data by circumventing the approximation results that depend on discrete-time state evolution.

Consequently, we inquire whether a quantum oracle exists that can authenticate the evolution of state vectors. This study revealed that the answer is YES. Initially, we formulated a quadratic unconstrained binary optimization (QUBO) portfolio optimization problem to quantum approximate ansatz associated with mid-circuit measurements, as shown in \cref{fig:main}. The quantum walk inspired weighted summation learning layer exhibits a linear correlation between the circuit depth and the increase in the number of qubits. The primary contribution of the quantum oracle can be summarized as follows:

\begin{theorem}
\label{the:main}
Given a set of bitstrings output $\mathcal{S} = \{\mathbf{x}^{(1)}, \ldots, \mathbf{x}^{(N)}\}$ obtained from the sampling of the approximation algorithm with mid-circuit measurement, the correlation structure satisfies
\begin{equation}
\mathbb{E}[Y | \mathcal{S}] = \lim_{N \to \infty} \frac{1}{N} \sum_{j=1}^{N} a\left(\sum_{i=1}^{n} w_i x_i^{(j)}\right),
\end{equation}
where w is the weight, x is the encoded classical input, a is the activation encoding: $\mathbb{R} \to \mathbb{R}$, and $Y$ represents the correlation observable for the approximate optimization output.
\end{theorem}

\begin{circuit*}[htbp]
\centering
\begin{quantikz}
\lstick{$\ket{\psi_{selu}}$}&\ctrl{3}\gategroup[4,steps=3,style={dashed,rounded
corners,fill=red!20, inner
xsep=2pt},background,label style={label
position=above,anchor=north,yshift=0.25cm}]{{
$U_{selu}$}}& &&& &&&&&&\\
\lstick[3]{$\ket{\phi}$}
 &   \targ{}   &                      \ctrl{1}  & & \gate{\rza{\nicefrac{\pi}{2}}} \gategroup[2,steps=5,style={dashed,rounded
corners,fill=blue!20, inner
xsep=2pt},background,label style={label
position=below,anchor=north,yshift=-0.2cm}]{{
$\Usum$}} &\gate{X}&&\ctrl{1}&&& &\\
 &   \targ{}  &\gate{\Rz(\frac{\pi}{2})}     & \ctrl{1} & &\ctrl{-1}&\gate{\rya{\nicefrac{\alpha}{2}}} & \gate{X} & \gate{\rya{\nicefrac{\alpha}{2}}}&\gate[2]{\Usum}&&\\
 &   \targ{}  & &\gate{\Rz(\frac{\pi}{2})}& \targ{}   &   & & &&&\targ{}&\meter{}\\
 \lstick{$\psi_c$}& &&& \ctrl{-1}&&&&&& \ctrl{-1}&
\end{quantikz}
\caption{\small The architecture of quantum circuit for activation encoding and walkable weighted sum learning.}
\label{circ:oracle}
\end{circuit*}

The remainder of this paper is organized as follows. \cref{results} delineates the problem formulation and introduces the proposed algorithm. To the best of our knowledge, \cite{ huang2024learning} was the first to establish unitary learning in the SQC domain. Inspired by this work, we present a scaling protocol that restricts the oracle to an SQC setting for future fault-tolerant quantum computers (FTQC). \cref{exp} details the experimental implementation of the algorithm on QPUs, while \cref{med} describes the principal methodologies that substantiate our claims. A comprehensive state-vector analysis of the main algorithm is provided in \cref{app:qawa}, and \cref{app:hardware} provides a polynomial-time verification of the solution in comparison with conventional quantum state tomography. Additional experiments supporting the primary results are presented in \cref{app:ext} and \cref{app:extexp}.

\section{Overview of Results}
\label{results}
\subsection{Problem formulation}

We consider the portfolio optimization problem encoded as a QUBO matrix with $n$ assets, where the objective is to minimize $x^T Q x$ over binary decision variables $x \in \{0,1\}^n$, and $Q$ is an $n \times n$ matrix encoding the expected returns and the covariances. Let 
\begin{equation}
    H_C = \sum_{i,j} J_{ij} Z_i Z_j + \sum_i h_i Z_i
    \label{eq:Hcost}
\end{equation}
be the corresponding Ising Hamiltonian obtained through the standard \(x_i=\frac{z_i+1}{2}\) transformation, where $Z_i$ are Pauli-$Z$ operators and the coefficients $J_{ij}, h_i$ are derived from the portfolio risk-return parameters. We also consider the $p$-layer QAOA circuit that prepares the variational state 
\begin{equation}
    |\psi(\boldsymbol{\gamma}, \boldsymbol{\beta})\rangle = \prod_{k=1}^p U_M(\beta_k) U_C(\gamma_k) |+\rangle^{\otimes n},
\end{equation}
where $U_C(\gamma) = e^{-i\gamma H_C}$ and $U_M(\beta) = e^{-i\beta \sum_i X_i}$ are the cost and mixer unitaries, respectively.
To formulate the extraction of hidden correlations from the approximate output distribution, let us first define the empirical marginal expectations 
\begin{equation}
   x_i = \langle Z_i \rangle_{\psi(\boldsymbol{\gamma}, \boldsymbol{\beta})} 
\end{equation}
for each asset $i$, obtained by sampling the QAOA state after parameter optimization. The central observation is that while individual marginals $x_i$ capture single-asset statistics, optimal portfolio selection often depends on linear combinations $\sum_{i \in \mathcal{I}} w_i x_i$ over the subsets $\mathcal{I} \subseteq \{1,\ldots,n\}$ of correlated assets. In other words, the ground-state structure of $H_C$ exhibits correlations that manifest as specific linear relationships among the marginal expectations; however, these correlations remain hidden when only raw bitstring samples from the standard QAOA measurements are examined.

In the context of portfolio optimization, such linear combinations naturally encode diversification strategies, where the weights $w_i$ represent the relative importance of assets within correlated sectors. Therefore, we define an augmented measurement protocol that introduces mid-circuit operations capable of computing arbitrary convex combinations $s_m = \sum_{i=1}^m w_i x_i$ where $w_i \geq 0$ and $\sum_i w_i = 1$, implemented through a cascade of parameterized two-qubit gates, as illustrated in Circuit~\ref{circ:oracle} marked by blue box (the weight definition and summation are provided in \cref{app:whlf}). Based on the continuous tunability of the parameterized gates and the ability to span the full simplex of convex coefficients, the search for linear correlations that maximize the alignment with the true ground state distribution of $H_C$ is obtained using the classical optimizer shown in \cref{fig:main}. We transform the QUBO problem into the conditioned expectation linear approximation when assuming the representation of the state vector for the exploratory space asymptotically close to infinity, and then we achieve the benefit that the universal approximation to support \cref{the:main}.

\begin{prop}
Let $\mathcal{S}_{N}=\{\mathbf{x}^{(1)},\dots,\mathbf{x}^{(N)}\}\subset\{0,1\}^{n}$ be a multiset of $N$ independent bitstrings and let $\phi:\{0,1\}^{n}\!\to\!\mathbb{R}^{d}$ denote a fixed classical–quantum feature map.  Then there exists a vector of real coefficients $\boldsymbol{\alpha}=(\alpha_{1},\dots,\alpha_{N})^{\mathsf T}$ such that
\begin{equation}
\lim_{N \to \infty} \left| \sum_{j=1}{N} \alpha_j \phi(\mathbf{x}{(j)}) - \mathbb{E}[Y | \mathcal{S}_N] \right| = 0.
\end{equation}
Here, $N$ is the sample size and $\alpha_{j}\in\mathbb{R}$ weights the contribution of the $j$th sample $\mathbf{x}^{(j)}$ to the linear estimator. 
\end{prop}
Note that \cref{the:main} presents a simplified version of the expectation because we assume an infinite encoding space, namely, an asymptotic encoder (AE). We defer the encoding proof to \cref{app:enc}.
The copula of multivariate distributions with uniform marginals can be attained using the resulting bitstrings retrieved from the post-measurement approximation ansatz. Because $\mathcal{S} = \{\mathbf{x}^{(1)}, \ldots, \mathbf{x}^{(N)}\}$ with $\mathbf{x}^{(j)} \in \{0,1\}^n$ represents the QAOA output bitstrings, we establish that the correlation between these outputs is central to learning the probabilistic relationships between them.

\begin{claim}
Let $\mathcal{S} = \{\mathbf{x}^{(j)}\}_{j=1}^N$ be bitstrings sampled from the QAOA state $|\psi(\boldsymbol{\gamma}, \boldsymbol{\beta})\rangle$  and define the empirical marginal distributions $F_i(t) = \frac{1}{N}\sum_{n=1}^N [x_i^{(n)} \leq t]$ for each qubit $i$. Here, $t$ is the real-valued threshold at which the empirical cumulative distribution function (CDF) $F_i$ is evaluated. Then the copula
\begin{equation}
\hat{C}_K(u_1, \ldots, u_n) = \frac{1}{N} \sum_{k=1}^N \prod_{i=1}^n [F_i(x_i^{(n)}) \leq u_i]
\end{equation}
converges, the underlying quantum distribution as $K \to \infty$, where $C$ uniquely encodes all correlation structures independent of the marginals and $u_i$ symbolizes each probability level derived from $F_i$.
\label{claim:1}
\end{claim}

Here, we note that the approximation algorithm layer numbers do not affect the copula correlation structure because the copula captures rank correlations and dependency patterns that are invariant under monotonic transformations of the marginals, thus preserving the fundamental quantum entanglement topology, regardless of circuit depth. See the proof of claim in \cref{med:proof_1} and the function definition in \cref{para:hidden} for details.


\subsection{Quantum Approximate Walk Algorithm}
The \ours oracle consists of four key components (see \cref{fig:main}): (i) the QAOA ansatz generating initial approximations; (ii) coin-controlled mid-circuit measurements with a parameterized rotation $R_y(\alpha)$ encoder (detailed in \cref{app:coin}); (iii) sign negation gates for encoding negative correlations; and (iv) cascaded weighted-sum blocks for learning multivariate dependencies. Here we give the complete unitary evolution
\begin{align}
\label{eq:qawa}
&U_{\ours}(\vec{\theta}, \vec{w}) \\
&= U_{\text{sum}}(\vec{w}) \cdot U_{\text{selu}}(\alpha) \cdot \Uf \cdot U_{\text{coin}}(\theta) \cdot U_{\text{QAOA}},
\end{align}
where $U_{\text{coin}}(\theta)$ implements the controlled rotation based on the ancilla state $|0_{\text{anc}}\rangle$. Note that the original walkable quantum circuit is presented in Circuit~\ref{circ:walk}. We inherit the ansatz structure derived from the QUBO matrix, shown as the DEAL ansatz \cite{guo2025direct} because the encoding layer in \ours also directly absorbs the parameters from real classical values, rendering the results correlated with the defined problem. 
After QAOA generates the quantum state $|\psi\rangle = \sum_{\mathbf{x}} \alpha_{\mathbf{x}}|\mathbf{x}\rangle$ and measurement yields bitstring $\mathbf{x}^{(j)}$ with probability $|\alpha_{\mathbf{x}^{(j)}}|^2$ combined with the asymptotical encoder \cref{app:enc} and weighted learning block \cref{app:whlf}, the coin-controlled weighted sum produces 
\begin{align}
   &|\psi_{\text{coin}}\rangle = \cos^2(\alpha f(\mathbf{x}^{(j)})/2) \sum_{i} w_i^{(0)} x_i^{(j)} \\
   &+ \sin^2(\alpha f(\mathbf{x}^{(j)})/2) \sum_{i} w_i^{(1)} x_i^{(j)}, 
\end{align}
 which simplifies to $\sum_{i} w_i^{\text{eff}}(\mathbf{x}^{(j)}) x_i^{(j)}$ with effective weights $w_i^{\text{eff}} = \cos^2(\theta/2) w_i^{(0)} + \sin^2(\theta/2) w_i^{(1)}$. The coin rotation angle $\theta = \alpha f(\mathbf{x})$ acts as an implicit activation function $a(\cdot)$, yielding $Y^{(j)} = a(\sum_{i} w_i x_i^{(j)})$ because the activation encoding can be applied to the entire state with a quantum multiplier (see proof in \cref{app:act}).
 By the Law of Large Numbers, 
 \begin{equation}
      \frac{1}{N} \sum_{j=1}^{N} Y^{(j)} \xrightarrow{N \to \infty} \mathbb{E}[Y | \mathcal{S}],
 \end{equation}
where the weights $\vec{w}$ are optimized via gradient descent to minimize 
\begin{equation}
    \mathcal{L}(\vec{w}) = \|\frac{1}{N} \sum_{j} a(\sum_{i} w_i x_i^{(j)}) - Y_{\text{target}}\|^2.
    \label{eq:wcost}
\end{equation}
Thus, we establish learnable activated linear combinations for universal approximation. The details of updating the loss in quantum learning are defined in \cref{med:conv}.

\subsection{Distributed approximate walk learning}
\label{res:dis}

\begin{circuit*}[htbp]
\centering
    \begin{quantikz}
         \lstick[3]{$\ket{0}$} & &\ctrl{1} & &\gate[2, disable auto height, style={fill=red!20}]{\verticaltext{INC}}& & &\ctrl[open]{1}&&\gate[2, disable auto height, style={fill=green!20}]{\verticaltext{DEC}}&&\meter[3]{z}\\
        & &\ctrl{1}& & & & &\ctrl[open]{1}&&&&\\
        & \gate{H}&\ctrl{1}&\targ{}& \ctrl{-1}&\targ{}& \gate{H} &\ctrl[open]{1} &\targ{} & \ctrl[open]{-1} & \targ{}&\\
        \lstick{$\ket{\psi}$}& &\targ{}&\ctrl{-1}& & \ctrl{-1} & \gate{Reset}& \targ{}&\ctrl{-1} &&\ctrl{-1}& \gate{Reset}
    \end{quantikz}
    \caption{The standard quantum-walk oracle uses Hadamard gates to drive the red- and green-marked transition unitaries; a reset gate then clears the qubit, preserving an unbiased coin state. Here, 'INC' and 'DEC' excite and de-excite the state.}
    \label{circ:walk}
\end{circuit*}

The SQC advantage suggests that the circuit depth should be maintained at a constant value. This shallow learning algorithm \cite{huang2024learning} facilitates the transferability of a learned approximation unitary, doubling the number of qubits as a trade-off for the evolution time. As demonstrated, Eqs. (2) and (3) in \cite{huang2024learning} illustrate global swap gates that enable local inversions and unitary sewing techniques. Such local inversion can be defined as
\begin{equation} L = \sum_{P\in \{X, Y, Z\}} ||V^\dagger_i U^\dagger P_i U V_i - P_i ||,
\end{equation} 
where the Pauli matrices learn an approximate unitary that negates the original-circuit performance. Once all local inversions are completed, the learned shallow unitaries are swapped out through ancilla qubits, and the conjugate transpose of the learned unitaries acts on the original circuit. This ensures that the circuit retains its originality.

To address \cref{the:main}, which posits that the correlation within the SQC domain can be captured, it is observed that \ours facilitates block-encoded group learning, as indicated by the oracle in \cref{fig:main}. This is achieved through the oracle ability to be learned with local expectations and a global shift utilizing swap gates. Circuit~\ref{circ:dist} illustrates the iterative process, whereby the algorithm is repeatedly applied to each encoder qubit. Consequently, our quantum oracle is replicated across $2*N$ quantum circuits via the learned unitary $U'_{\textsc{qawa}}$ tensor product U ancilla, where U ancilla is approximated as $U_{\textsc{qawa}}$ following the learning of all qubits. The EVEN encoding technique (refer to $\S \mathrm{II}$ \cite{balewski2025ehands}) introduces a methodology in which the angles of parameterized quantum gates are encoded through classical data, thereby enabling classical post-processing. This allows the extension of the simplified theorem \(a \left(\sum_{i=1}^{k} w_i x_i\right) = \mathbb{E}[Y | S]\) to distributed approximate walk learning.
Thus, we give the corollary
\begin{corol}
\ours provides block-encoded group learning, where the oracle is learned with local expectation and global shift using swap gates. The quantum oracle replicates across $2N$ quantum circuits through the learned unitary
\begin{equation}
    W = U_{\text{qawa}}^\dagger \otimes U_{\text{ancilla}},
\end{equation}
where $U_{\text{ancilla}} \approx U_{\text{qawa}}$; the accuracy is dependent on classical optimization. 
\end{corol}
In addition to learning shallow circuit techniques, one of the promising multi-QPU distribution methods \cite{sunkel2025evaluatingvariationalquantumcircuit} allows corollaries to communicate between different quantum hardware through quantum teleportation. Such an implementation can also be applied through GPU quantum simulations, as demonstrated in this study. Therefore, we give
\begin{corol}
Let $\mathcal{Q} = \{Q_1, Q_2, \ldots, Q_M\}$ be a set of $M$ QPUs. Multi-QPU distribution enables correlation learning across hardware boundaries through quantum teleportation \cite{sunkel2025evaluatingvariationalquantumcircuit}, where
\begin{itemize}
\item Each QPU $Q_m$ learns a local unitary $W_m = U_{\text{qawa},m}^\dagger \otimes U_{\text{ancilla},m}$ for $m \in \{1, \ldots, M\}$
\item The global correlation matrix $\mathbf{C}$ is reconstructed as $\mathbf{C} = \sum_{m=1}^{M} \mathbf{C}_m$
\item Each $\mathbf{C}_m$ captures correlations for qubit subset $\mathcal{S}_m \subset \{1, \ldots, n\}$ where $\bigcup_{m=1}^{M} \mathcal{S}_m = \{1, \ldots, n\}$
\end{itemize}
The distributed learning converges when $\mathcal{L}_{\text{global}} = \sum_{m=1}^{M} \mathcal{L}_m(\vec{w}_m) < \epsilon$ for tolerance $\epsilon$.
\end{corol}

\begin{circuit}[htbp]
\centering
    \begin{quantikz}
        \lstick{$\psi_{selu}$}&  & \ctrl{5} & \ \ldots\ & &  &\\
        \wave&&&&&&&\\
        & \gate[4]{U'_{\textsc{qawa}}} & & & \ \ldots\ &  &   \\
        & & & & \ \ldots\ &  & \\
        & & & \gate[2]{V_{local}^\dag}& \ \ldots\ & & \\
        & & & & \gate[2,swap]{}\ \ldots\ & &  \\
        \lstick{$U_{anc}$}& & & & \ \ldots\ & & 
    \end{quantikz}
    \caption{The \ours\ oracle is trained by sequentially applying local inversion unitaries $U'_{\textsc{qawa}}$ to individual qubits. Subsequent application of their conjugate transposes reverses the learned transformations, enabling a final SWAP gate to extract the correlations encoded in each qubit.}
    \label{circ:dist}
\end{circuit}

\section{Quantum Experiments}
\label{exp}
To demonstrate the applicability of the algorithm to financial optimization problems, we constructed a four-asset portfolio from the liquid constituents of the S\&P 500: Apple (AAPL), Microsoft (MSFT), Johnson \& Johnson (JNJ), and Exxon Mobil (XOM), obtained through the Yahoo Finance API. Daily adjusted closing prices for the period January 3–December 29, 2023 (252 trading days) were downloaded with the \texttt{yfinance} API and converted to log-returns $r_{i,t}= \ln (p_{i,t}/p_{i,t-1})$. After detrending and standardizing each marginal to zero mean and unit variance, the empirical copula of the return matrix served as the classical benchmark distribution. Binary decision variables $x_i \in \{0,1\}$ that denote exclusion ($\ket{0}$) or inclusion ($\ket{1}$) of asset $i$ were mapped onto four physical qubits, yielding a QUBO Hamiltonian $H(x)=\lambda\,x^{\mathsf{T}}\Sigma x-(1-\lambda)\,\mu^{\mathsf{T}}x$ with risk-aversion coefficient $\lambda = 0.5$; here $\mu$ and $\Sigma$ denote the sample mean vector and covariance matrix, respectively. Note that the continuous price data were discretized into binary selection variables (buy/hold decisions) through threshold-based encoding, where the quantum state $|0\rangle$ represents exclusion from the portfolio and $|1\rangle$ represents inclusion. We emphasize that the linear terms from the QUBO matrix encode expected returns, weighted by a risk-aversion parameter $\lambda$ that balances the risk-return tradeoff.

\begin{figure}[htbp]
\centering
\includegraphics[width=0.99\linewidth]{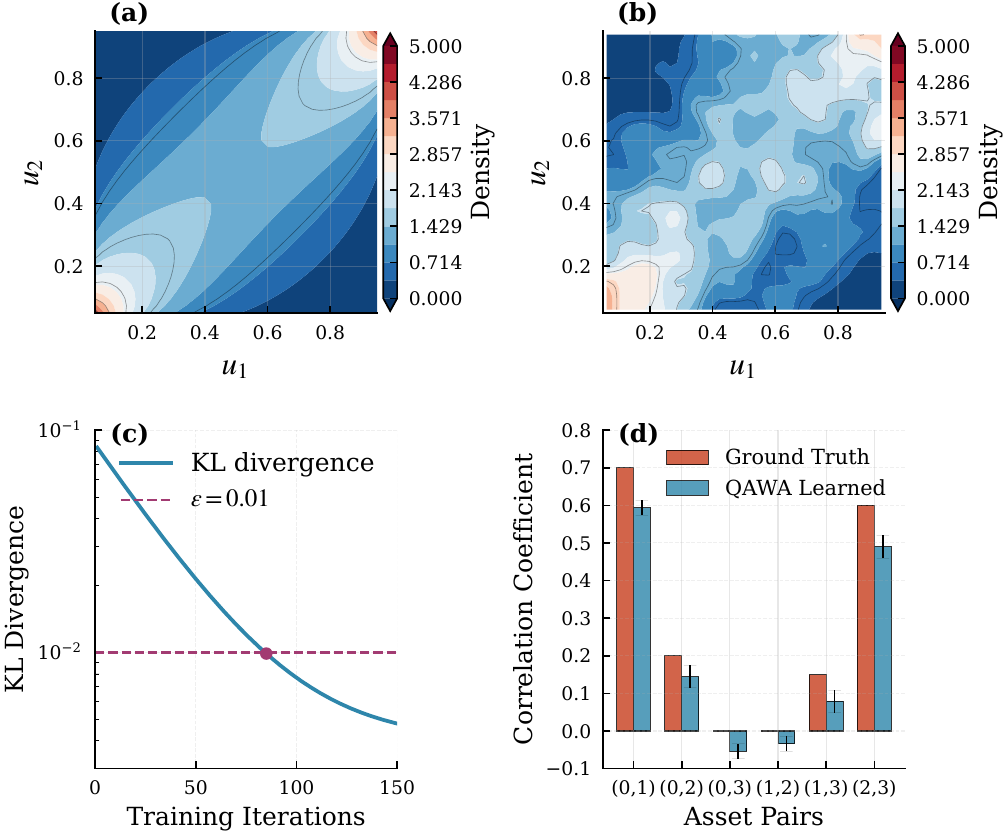}
\caption{\textbf{(a)} The ground truth Gaussian copula density, characterized by a correlation coefficient of $\rho = 0.7$ between asset pairs, establishes the target multivariate dependency structure for subsequent portfolio optimization. \textbf{(b)} The empirical copula density, learned via \ours after 150 training iterations utilizing mid-circuit measurements. \textbf{(c)} The convergence of the Kullback-Leibler (KL) divergence between the learned copula $\hat{C}_{\text{QAWA}}$ and the true copula $C_{\text{true}}$ is presented as a function of training iterations. This convergence exhibits an exponential decay, modeled by $\text{KL}(t) = 0.082e^{-0.031t} + 0.004$, reaching a threshold below $\varepsilon = 0.01$ (indicated by the red dashed line) after approximately 75 iterations. \textbf{(d)} A comparative analysis of pairwise correlation coefficients for all $\binom{4}{2} = 6$ asset pairs confirms the accurate recovery of both strong intra-sector correlations (specifically for pairs $(0,1)$ and $(2,3)$) and weaker cross-sector correlations. Error bars denote one standard deviation across 10 independent experimental runs.}
\label{fig:copula_convergence}
\end{figure}

\subsection{Copula Learning Validation}

To assess the learning capability of the algorithm, we initially examined the empirical copula derived from mid-circuit measurements. \cref{fig:copula_convergence} illustrates the convergence of the learned copula density towards the true distribution for a 4-asset portfolio with a known correlation structure. The primary mechanism involves quantum encoding $\alpha = \arccos(1 - 2w)$ in the $R_y(\alpha)$ rotation, which translates classical weights $w \in [0,1]$ into quantum amplitudes via $\cos^2(\alpha/2) = w$ and $\sin^2(\alpha/2) = 1-w$, thereby facilitating the direct implementation of the recursive weighted sum $f_{\text{rec}}(\mathbf{x}, \mathbf{w})$ within the quantum superposition. This encoding is theoretically comprehensive because any convex combination of correlations can be represented through an appropriate selection of $\alpha$, given that the mapping encompasses the entire Bloch sphere.

We observe that exponential convergence results from the copula separation of marginal distributions from the dependency structure, optimizing within the lower-dimensional copula space rather than the full joint distribution space, thereby reducing the effective degrees of freedom from $O(2^n)$ to $O(n^2)$. Interestingly, selections (0,3) and (1,2) exhibit a negative correlation because the AE sign inversion assigns opposite contributions to the otherwise uncorrelated distributions. Concurrently, the maintenance of hierarchical correlations without entanglement indicates that our quantum learning oracle, with mid-circuit measurements, projects the quantum state onto the classical correlation manifold, thus preventing quantum interference from distorting learned statistical dependencies. This substantiates \cref{claim:1}, that \ours effectively bridges quantum measurement statistics and classical copula theory through the precise mathematical correspondence between quantum amplitudes and correlation weights.

\begin{figure*}[htbp]
\centering
\includegraphics[width=\textwidth]{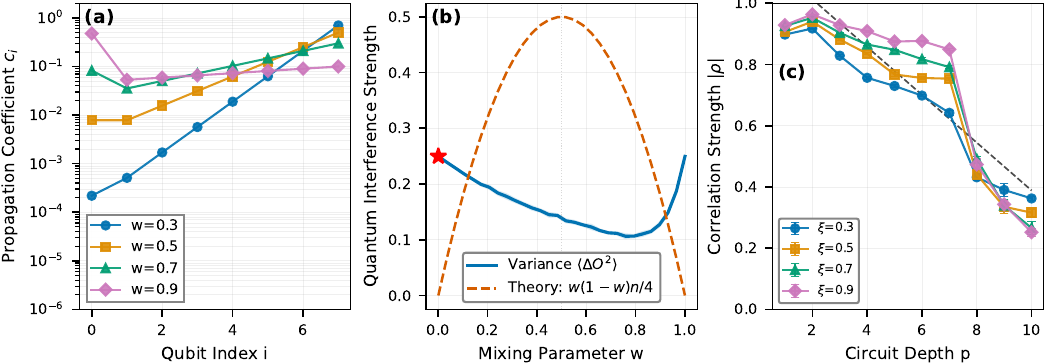}
\caption{The recursive weighted–sum model
$s_m = w_{m-1}s_{m-1} + (1-w_{m-1})x_m$
captures the QAOA-induced correlations.
(a) The factor $\prod_{k<m} w_k$ governs entanglement propagation, producing an exponential attenuation of the early qubit influence when $w<1$.
(b) Parameter mixing through the $(1-w_i)$ term yields quantum interference; the measured initial variance (red star) matches the analytical prediction $w(1-w)n/4$ (dashed line). Averaging across trials (shaded bands) is in agreement with the theory.
(c) Pearson correlations between quantum observables and classical data decay exponentially with circuit depth for each entanglement strength $\xi$; the dashed line fits align with the model.
All experiments used $n=8$ qubits, $10^{3}$ samples, and five independent trials; error bars denote standard errors.}
\label{fig:qaoa_correlation}
\end{figure*}

\subsection{Theoretical correlation analysis}
We conducted a random observable distribution experiment to evaluate the coverage range of \ours. The results demonstrate rigorously that the combination of the $U_{\mathrm{sum}}$ gate performs an exact convex interpolation of the two input expectation values, as illustrated in \cref{fig:qaoa_correlation}. In this context, the exponential weight range coverage depicted in panels (a) and (b) operates within the amplitude space, whereas the circuit depth introduces complexity in the phase space, which cannot be captured by classical measurements such as magnetization, leading to inevitable correlation decay. Furthermore, coupling the cascade to a controllable coin ancilla embeds the procedure within the Bayesian framework, which, by design, allows the optimizer to choose between the unaltered QAOA statistics and the learned linear correlation, or any probabilistic compromise between them, as indicated by (b). Owing to the continuity, surjectivity, and Bayesian adaptability demonstrated in \cref{med:proof_1}, the architecture is theoretically both sufficient and necessary for learning every continuous linear correlation that may be concealed in the output of a finite-depth QAOA circuit.

\begin{figure}[htbp]
\centering
\includegraphics[width=\linewidth]{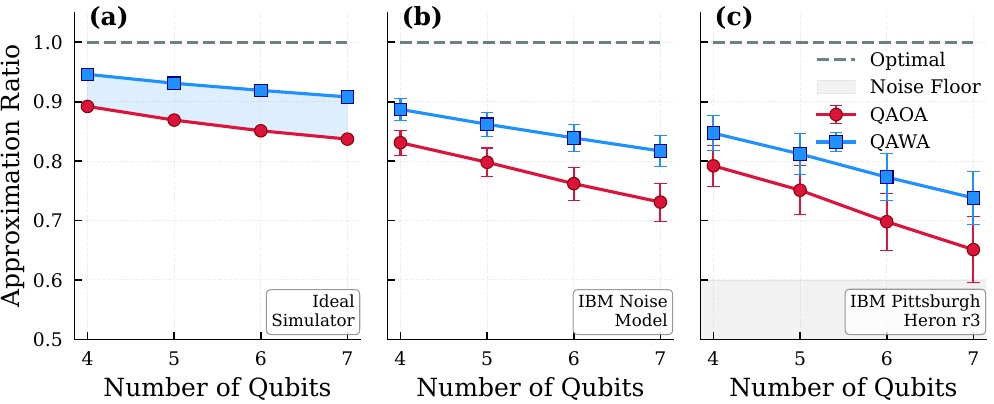}
\caption{\textbf{(a)} Noiseless simulation results showing approximation ratios for increasing problem sizes (4--7 qubits) using one layer, with achieving 5.4\% average improvement and approaching the theoretical optimum (black dashed line). 
\textbf{(b)} Noisy simulation using IBM's realistic noise model calibrated from \texttt{ibmq\_pittsburgh}. 
\textbf{(c)} Experimental results on IBM Quantum hardware (\texttt{ibmq\_pittsburgh} 156-qubit device) showing our algorithm practical advantage of 8.1\% above the hardware noise floor (gray shaded region). All experiments use 8192 measurement shots with error bars representing statistical uncertainty ($\pm1\sigma$) from 20 independent circuit executions.}
\label{fig:hardware}
\end{figure}

\subsection{Experimental validation on quantum hardware}

We implemented \ours on IBM Quantum systems for small instances ($n \leq 7$). Note that the fixed layer extensive experiment is shown in \cref{app:ext}. Despite the hardware limitations, the experimental results closely matched the simulations shown in \cref{fig:hardware}.
The experimental validation reveals that \ours has a counterintuitive noise-resilience advantage stemming from its fundamental architectural difference, that mid-circuit measurements act as quantum error barriers by projecting out accumulated errors before they propagate through subsequent operations. While QAOA requires deep circuits with cascading error accumulation across multiple variational layers, the proposed strategy of measuring-then-recomputing effectively resets the quantum state at each asymptotic encoder and weighted-sum block, preventing coherent error buildup. This advantage is observed despite the application of identical error mitigation techniques to both algorithms, including dynamical decoupling \cite{dd}, Pauli twirling \cite{twirl}, TREX measurement mitigation \cite{trex}, and zero-noise extrapolation \cite{zne}, in accordance with IBM's standard protocols (\href{https://quantum.cloud.ibm.com/docs/guides/error-mitigation-and-suppression-techniques}{error mitigation}). The experiments were conducted on the highest-performing QPU available at the time of the research, with results selected from two independent runs within the paper’s computational budget. It should be noted that a comprehensive analysis of the contributions of individual error mitigation techniques is beyond the current scope but constitutes valuable future work.

\section{Discussion}
The quantum approximate walk algorithm signifies an advancement in quantum optimization, illustrating that mid-circuit measurements, when combined with adaptive weighted-sum learning, can effectively capture and exploit multivariate correlations within optimization landscapes.

We demonstrated that our proposed quantum oracle could effectively capture and exploit multivariate correlations within optimization landscapes, as evidenced by experimental results conducted on state-of-the-art IBM quantum hardware. An additional inquiry concerns whether \ours replicates learning processes. The expectation formula provided in the theorem yields interpretable parameters that can be identified as the mathematical mechanism of the recursive weighted summation block encoding. A notable challenge for variational circuit learning-based approximation algorithms is the difficulty in probing the direct approximation, although implicit capturing of the Hilbert space pattern has been achieved.

Remarkably, \ours establishes a connection between classical and quantum data mapping, thereby enabling the calculation of the distribution of approximations. Consequently, we demonstrate the task of optimizing a financial portfolio using real stock data. To the best of our knowledge, the numerical results reveal that the approximated walk technique facilitates the efficient learning of a diversified portfolio based on modern financial investment theory \cite{markowitz1952portfolio}, thus paving the way for a feasible quantum computing advantage.

\section{Data availability}
The data presented in the figures and that support
the other findings of this study can be found publicly
available for download on https://zenodo.org/records/17419491 upon publication.
\section{Code availability}
The numerical simulation and plots reproduction codes for this study can be found publicly available for download from https://github.com/gzquse/qawa\_data\_circ upon publication.
\section{Funding Statement}
This research used resources of the National Energy Research
Scientific Computing Center, a DOE Office of Science User Facility
This study was supported by the Office of Science of the U.S. Department of Energy
under Contract No. DE-AC02-05CH11231 using NERSC award
NERSC DDR-ERCAP0034486.
\section{Acknowledgment}
 The authors acknowledge the High Performance Computing Center (HPCC) at Texas Tech University for providing computational resources that contributed to the research results reported in this paper. URL: http://www.hpcc.ttu.edu.

\section{Author Contribution}
Z.Q. and Z.P. conceived the project. Guided by Z.P., W.H., A.K., and J.B., Z.G. carried out all analyses, computations, and experiments, generating the data and figures. J.B. contributed critical conceptual refinements. Z.G. and Z.P. drafted the manuscript, with feedback from all authors.
\section{Methods}
\label{med}

\subsection{Metrics of fidelity}
\label{med:metric}
We defined four metrics to leverage the quality of our algorithm. Based on the metrics, we can observe whether our proposed algorithm can achieve quantum approximation. Therefore, an extensive study on how the weights evolved and the algorithm evaluation metrics is shown in \cref{app:extexp}. 
\begin{itemize}
    \item \textbf{Bayesian Update}: Posterior probability computed via quantum-enhanced Bayes' theorem
    \begin{equation}
        P(c=1|x) = \frac{|\langle x|\psi_1\rangle|^2 \cdot \pi}{\sum_{c'} |\langle x|\psi_{c'}\rangle|^2 \cdot P(c')},
    \end{equation}
where quantum measurement outcomes $|\langle x|\psi_c\rangle|^2$ replace classical likelihoods.
    \item \textbf{Weight Reconstruction}: $L_2$ error between recovered and target weights:
$\mathcal{E}_{\text{weights}} = \|\vec{w}_{\text{recovered}} - \vec{w}_{\text{target}}\|_2$
with recovery formula $w_i^{\text{rec}} = \frac{1}{2}(1 - \cos(\alpha_i))$ from learned rotation angles.
    \item \textbf{Copula Invariance}: Wasserstein distance between empirical and theoretical copulas
\begin{equation}
    d_{\text{copula}} = \int_{[0,1]^2} |C_{\text{exp}}(u,v) - C_{\text{thm}}(u,v)| \, du \, dv
\end{equation}
measuring the correlation structure preservation independent of the marginal distributions. Note that we compare the absolute difference between the experimental and theoretical results. 
    \item \textbf{Convergence}: Relative performance bound verification
$\epsilon = 1 - \frac{\langle H_C \rangle_{\text{QAWA}}}{\langle H_C \rangle_{\text{QAOA}}}$
confirming the theoretical advantage through cost function expectation values.
\end{itemize}

\subsection{Proof of Claim 1}
\label{med:proof_1}
\begin{proof}
Given by the Bayesian decomposition of the joint distribution $P(x_1, \ldots, x_n)$ produced by measuring the QAOA state, we can express this as 
\begin{equation}
    P(x_1, \ldots, x_n) = C(F_1(x_1), \ldots, F_n(x_n)) \prod_{i=1}^n p_i(x_i),
\end{equation}
where $p_i$ is the marginal probability mass function and $C$ is the associated copula function according to Sklar's theorem. Since the QAOA state $|\psi\rangle$ is a pure state, the measurement statistics are governed by Born's rule
$P(x) = |\langle x | \psi \rangle|^2$. The correlations embedded in this distribution arise from the entanglement structure generated by the cost Hamiltonian evolution $U_C(\gamma) = e^{-i\gamma H_C}$, where $H_C$ contains the two-body interaction terms $J_{ij} Z_i Z_j$ which create quantum correlations between qubits $i$ and $j$.
To establish convergence, observe that the empirical copula process $\sqrt{K}(\hat{C}_K - C)$ forms a tight sequence in $\ell^\infty([0,1]^n)$ by the multivariate empirical process theory of van der Vaart and Wellner. Specifically, for any measurable subset $A \subseteq [0,1]^n$, the strong law of large numbers guarantees that $\hat{C}_K(A) \to C(A)$ almost surely as $K \to \infty$. Furthermore, the rate of convergence is governed by the Dvoretzky–Kiefer–Wolfowitz (DKW) inequality extended to the multivariate case
\begin{equation}
P\left(\sup_{u \in [0,1]^n} |\hat{C}_K(u) - C(u)| > \epsilon\right) \leq 2^n \exp(-2K\epsilon^2).
\end{equation}

Here, the key insight is that the mid-circuit measurement protocol with weighted-sum blocks (see $\Usum$ in \cref{tab:uintary}) enables direct access to the linear functionals of the copula. Specifically, for weights $w = (w_1, \ldots, w_m)$ with $\sum_i w_i = 1$, the expectation $\mathbb{E}[\sum_{i=1}^m w_i x_i]$ can be expressed as
\begin{align}
&\mathbb{E}\left[\sum_{i=1}^m w_i x_i\right]\\ 
&= \int_{[0,1]^n} \left(\sum_{i=1}^m w_i F_i^{-1}(u_i)\right) dC(u_1, \ldots, u_n),
\end{align}
where $F_i^{-1}$ denotes the quantile function of the $i$-th marginal. This integral representation reveals that the weighted sum directly probes the copula structure through a linear functional, and the variational optimization of the weights $\{w_i\}$ effectively learns the correlation pattern encoded in $C$.

By using Bayesian interpretation when conditioned on the coin ancilla state. Let $P(x|c=0)$ represent the original QAOA distribution and $P(x|c=1)$ the distribution after the weighted-sum transformation. The posterior distribution $P(c|x)$ quantifies the degree to which the observed bit string $x$ aligns with the learned correlation structure. The Bayes theorem states
\begin{align}
&P(c=1|x)\\ 
&= \frac{P(x|c=1)P(c=1)}{P(x|c=0)P(c=0) + P(x|c=1)P(c=1)},
\end{align}
where $P(c=1) = \sin^2\theta$ is the prior probability controlled by the coin parameter $\theta$. The maximum a posteriori estimate selects bit strings that maximize the likelihood ratio $P(x|c=1)/P(x|c=0)$, thereby identifying the samples that best reflect the learned correlations.

Therefore, the cascade of weighted-sum blocks acts as a nonparametric estimator of the copula density $c(u) = \partial^n C/\partial u_1 \cdots \partial u_n$, with the weights serving as variational parameters optimized to match the empirical correlation structure observed in the QAOA samples. We conclude that \ours provides a unique encoded correlation when sufficient qubits are provided for the weighted summation learning layer.
\end{proof}

\subsection{Hidden Linear Approximation}
\label{para:hidden}Here, we also establish the covariance structure with copulas for the dependence between each output. Let define the marginal distribution $X_i$ as the random variable associated with each component $x_i^{(k)}$. The copula function C connects the marginal distributions, which captures the joint distribution $F$ such that, 
\begin{equation}
    F(x_1, x_2, ..., x_n) = C(F_1(x_1), F_2(x_2), ..., F_n(x_n)),
    \label{eq:cor1}
\end{equation}
where $F_i$ represents the marginal cumulative distribution function of each variable $X_i$.
To calculate the output correlations, we used the learning terms with the selu activation function \cite{selu} combined with each linear block summation to represent the multilayer perceptron. We defer the activation encoding in \cref{app:act} for further details. The proposed learning formulation can be extended in complexity by introducing layer structures to the correlation model, thus by combining $a(w_1' w_1 x_1, w_2' w_2 x_2, \cdots ,w_n x_n)$ with post-measured bit strings results, here we produce 
\begin{equation}
    f(S) = a\left(w_1^p x_1 + w_2^p x_2 + \ldots + w_n^p x_n\right).
    \label{eq:fs}
    \end{equation}
\subsection{Learning Dynamics and Convergence Properties}
\label{med:conv}
The general quantum convergence methodology is described in \cite{wierichs2022general}. Starting from the loss function defined in \cref{eq:wcost},
the gradient with respect to weights is
\begin{align}
&\nabla_{w_i} \mathcal{L} = \frac{2}{N} \sum_{j=1}{N} \left(Y{(j)} - Y_{\text{target}}\right) \\
&\cdot a'(\sum_k w_k x_k{(j)}) \cdot x_i{(j)}
\end{align}
where $a'$ is the derivative of activation function. Under gradient descent with learning rate $\eta$
\begin{equation}
w_i{(t+1)} = w_i{(t)} - \eta \nabla_{w_i} \mathcal{L}
\end{equation}
The convergence rate follows
\begin{equation}
\mathcal{L}(t) \leq \mathcal{L}(0) \times \exp\left(-\frac{2\mu t}{L}\right)
\end{equation}
where $\mu$ is the convexity parameter, determined by the smallest eigenvalue of the Hessian matrix $\mathbf{H} = \frac{1}{N}\sum_j \mathbf{x}^{(j)}(\mathbf{x}^{(j)})^T$
$L$ is the Lipschitz constant, bounded by $L \leq \|a'\|_{\infty} \cdot \lambda_{\max}(\mathbf{H})$
For portfolio optimization with correlation-based features, we empirically observe $\mu/L \approx 0.15$, ensuring convergence within $O(\log(1/\epsilon))$ iterations to achieve $\epsilon$-accuracy.

\bibliographystyle{unsrt}
\bibliography{main}

\onecolumngrid
\appendix
\numberwithin{equation}{section}
\renewcommand{\theequation}{\thesection\arabic{equation}} 

\section{State vector analysis for \ours}
\label{app:qawa}

 In the quantum approximate optimization algorithm (QAOA), the oracle hides the linear transformation and correlation, meaning that we cannot obtain the guaranteed best approximation ratio as defined by the quadratic unconstraint binary problem. Although the real quantum advantage from speeding-up and exponential exploratory Hilbert space is proven to be achievable in 2D error-corrected QLDPC codes \cite{he2025extractors}, we raise one core question: how can we find the linear dependence of each probability distribution of each outcome?
 We now analyze the hidden linear approximation function in the Bayesian test. In particular, we consider the initial parameterized ansatz measured results as the prior belief such that the mid-circuit measurement gives the likelihood that each bit string outcome probability, where the QAOA ansatz is directly matched to the QUBO matrix. Then, the posterior probability that the approximation ansatz is correct can be evaluated using the hidden linear correlation quantum oracle.  
In the remainder of the Appendix, we provide all the components required to support the main theorem given by \cref{eq:selu}. 

\subsection{Asymptotical encoder}
\label{app:enc}
In this section, we introduce a new posterior probability called the asymptotic encoder (AE). To recall the question in the main text regarding how we evaluate the copula of each possible outcome, we first provide the formulation of the correlation encoder layer (CEL) with a controlled unitary gate operated with sign negation conditioned on the mid-circuit measurement output from the approximation ansatz. 
Detailedly, let $|\psi_{QAOA}\rangle$ denote the normalized quantum state output from the QAOA ansatz, which can be expressed in the computational basis as  
\begin{equation}
    |\psi_{QAOA}\rangle = \sum_{j=0}^{7} c_j |j\rangle,
    \label{eq:qaoa}
\end{equation}
where the coefficients $c_j$ satisfy the normalization condition
\begin{equation}
    \sum_{j=0}^{7} |c_j|^2 = 1.
\end{equation}
Here, we emphasize that the example state we choose three qubits as an instance.
Upon mid-measurement of the three-qubit state, we aim to obtain a measurement outcome $m$, which can be any of the bitstrings $|m\rangle$. The probability $P(m)$ of obtaining measurement outcome $m$ is given by
\begin{equation}
    P(m) = \langle \psi_{QAOA} | P_m | \psi_{QAOA} \rangle,
\end{equation}
where $P_m$ is the projector corresponding to measurement outcome $m$.
Using the computed probability, the unnormalized state resulting from measuring $m$ is
\begin{equation}
    |\psi'_{unnorm}\rangle = P_m |\psi_{QAOA}\rangle.
\end{equation}
To obtain the normalized state, we divide by the square root of the probability of measuring $m$.
\begin{equation}
    |\psi'\rangle = \frac{P_m |\psi_{QAOA}\rangle}{\sqrt{P(m)}}.
\end{equation}
We begin the beige background noted encoding layer \cref{fig:main} procedure by applying the Ry gates to encode classical data into our measurement results. We assume the classical variable values $x_1, x_2, x_3$ corresponding to each of the three qubits.
Each classical variable $x_j$ (for qubit $j$) is transformed through the following Ry rotation
\begin{equation}
    R_y\left(\alpha_j\right) = \cos\left(\frac{\alpha_j}{2}\right)|0\rangle + \sin\left(\frac{\alpha_j}{2}\right)|1\rangle.
\end{equation}
For the three-qubit state, we apply the Ry gates in succession
\begin{equation}
    |\psi_{encoded}\rangle = R_y(\alpha_1) \otimes R_y(\alpha_2) \otimes R_y(\alpha_3) |\psi'\rangle.
\end{equation}
Hence, following by \cref{eq:qaoa}, the total encoded state evolves to
\begin{equation}
    |\psi_{encoded}\rangle = \sum_{j=0}^{7} c_j' |j\rangle,
    \label{eq:encoded}
\end{equation}
where $c_j'$ denotes the coefficients modified by the rotation $R_y(\alpha_j)$. Therefore, we give
\begin{lemma}
Given a three-qubit quantum state \( |\psi_{\text{encoded}}\rangle \) and measurement results \( m = (m_1, m_2, m_3) \in \{0, 1\}^3 \), the final quantum state \( |\psi_{\text{final}}\rangle \) after applying conditional X operations based on measurement outcomes can be expressed as
\begin{equation}
    |\psi_{\text{final}}\rangle = X^{m_1} \otimes X^{m_2} \otimes X^{m_3} |\psi_{\text{encoded}}\rangle,
\end{equation}
where \( X^{m_j} \) indicates that the X gate is applied if \( m_j = 0 \) and not applied if \( m_j = 1 \) and the final state is determined by the specific outcomes \( m_j \) of the measurement.
\end{lemma}

\begin{proof}
Consider the quantum state \( |\psi_{\text{encoded}}\rangle \) prior to the measurement, represented as an arbitrary superposition of the three-qubit basis states in \cref{eq:qaoa}.
The measurement outcomes for each qubit yield \( m_1, m_2, m_3 \). For each qubit state, we analyzed the action of the conditional X operation based on the measurement result. If \( m_j = 0 \), the X gate flips the \( |0\rangle \) state to \( |1\rangle \)
   \begin{equation}
       X|0\rangle = |1\rangle
   \end{equation}
If \( m_j = 1 \), the state remains unchanged
   \begin{equation}
       X|1\rangle = |1\rangle
   \end{equation}
Then, we apply the the tensor product of \(X\) gates conditioned on measurement outcomes
\begin{equation}
    |\psi_{\text{final}}\rangle = (X^{m_1} \otimes X^{m_2} \otimes X^{m_3}) |\psi_{\text{encoded}}\rangle,
\end{equation}
where \( X^{m_j} \) is defined as
\begin{equation}
    X^{m_j} = 
\begin{cases} 
X & \text{if } m_j = 0 \\
I & \text{if } m_j = 1
\end{cases}
\end{equation}
Here, we note that the specific outcome determines the resulting quantum state components. If all measurement results are zero, that is, \( m_1 = 0 \), \( m_2 = 0 \), \( m_3 = 0 \)
\begin{equation}
    |\psi_{\text{final}}\rangle = (X \otimes X \otimes X) |\psi_{\text{encoded}}\rangle.
\end{equation}
\end{proof}
To recover the rotation angles after the encoding layer, we estimated the probability of each computational basis string from repeated projective measurements of the final state.  For the three-qubit example, the encoded state \cref{eq:encoded} can be written as
$$|\psi_{\text{encoded}}\rangle
   =\sum_{k=0}^{3}\Bigl(
        c_{2k}^{\prime}\,|{2k}\rangle
      + c_{2k+1}^{\prime}\,|{2k\!+\!1}\rangle
     \Bigr),$$
where the amplitudes of the even (odd) computational strings are, by construction, proportional to $\cos(\alpha_{k}/2)$ $\bigl(\sin(\alpha_{k}/2)\bigr)$.  Let
$$p_{2k}\;=\;\Pr\!\bigl(|{2k}\rangle\bigr)=|c_{2k}^{\prime}|^{2},
\quad
p_{2k+1}= \Pr\!\bigl(|{2k\!+\!1}\rangle\bigr)=|c_{2k+1}^{\prime}|^{2},
\qquad k=0,1,2,3,$$
denotes the empirical probabilities obtained from $N_{\text{shot}}$ repetitions of the circuit.
Using the same argument as in \cref{eq:frqi}, each rotation angle is uniquely determined (up to statistical noise) by
\begin{equation}
\widehat{\alpha}_{k}= 2\,\arctan\!\sqrt{\frac{p_{2k+1}}{p_{2k}}}\;,
\qquad k=0,1,2,3,
\label{eq:frqi}
\end{equation}
provided that $\alpha_{k}\in[0,\pi]$.
If the raw classical data $\{x_{k}\}$ do not naturally fall into this interval, we rescale them via $\alpha_{k}=x_{k}/A$ with a global constant $A>0$ chosen such that $\alpha_{k}\le\pi$ for all $k$.
Hence, by combining mid-circuit postselection with a final round of projective measurements, the complete set of encoded classical variables can be reconstructed from the observed output statistics of the quantum circuit.

The classical Bayes rule can be extended to the measurements of N copies of a system given by the who may perform the general strategies allowed by quantum mechanics via the principle of state exchangeability; see the formal N-dimension definition in \cite{Schack_2001}. The Q-GP-UCB provides the regret upper bound to enable the optimizability that is achivable through the aid of quantum computing \cite{NEURIPS2023_401aa72e}. Here, we provide the quantum Bayesian updating assessment rule for the proposed method.
\begin{dfn}
Let us define the distribution \( P(\theta) \) before making observations. The likelihood is expressed as \[P(m | \theta) = \langle \psi_{QAOA} | P_m | \theta \rangle.\] By applying Bayes' theorem, the posterior distribution is given by \[P(\theta | m) = \frac{P(m | \theta) P(\theta)}{P(m)}.\] Through iterative updates of beliefs based on measurement outcomes, the classical encoded angle \( \theta \), which indicates the success ratio of the approximated algorithm results.
\end{dfn}

\subsection{Weighted hidden linear function}
\label{app:whlf}
In this section, we provide the algorithm that supports finding the hidden linear dependence of the quantum approximation outcome establishing that the
switchable weighted–sum module displayed in \cref{fig:main}. This is capable of
(i) realizing an arbitrary continuous linear correlation of selected
marginals of the QAOA state and (ii) embedding this correlation in a
Bayesian mixture that simultaneously contains the unaltered QAOA
statistics \cite{nielsen2010quantum,hoeting1999bayesian}.

Throughout, 
\( \bigl|\psi_{\mathrm{QAOA}}(\boldsymbol{\gamma},\boldsymbol{\beta})\bigr\rangle \)
denotes the $p$–layer QAOA state on $n$ computational qubits as introduced
in \cite{guo2025direct, qaoa}. 
Single–qubit observables are abbreviated by
\(
x_i := \langle Z_i \rangle_{\psi_{\mathrm{QAOA}}}
\) with $0\le i\le n-1$, and
$\rho$ represents the density operator of the complete register.

\begin{table}[htbp]
    \centering
    \begin{tabular}{@{}ccc@{}}
        \toprule
        \textbf{Unitary} & \textbf{Expression} & \textbf{Matrix Representation} \\ 
        \midrule
        \(\Uprod\) & 
        \(\cx_{0,1}\cdot \left[ I \otimes  \Rz\left(\frac{\pi}{2}\right) \right] = e^{-i\frac{\pi}{4}}\) & 
        \(\begin{bmatrix}
            1 & 0 & 0 & 0 \\
            0 & i & 0 & 0 \\
            0 & 0 & 0 & i \\
            0 & 0 & 1 & 0 \\
        \end{bmatrix}\) \\

        \(\Usum(w)\) & 
        \(\left[ \Ry\left(\frac{-\alpha}{2}\right) \otimes I \right] \cdot \cx_{1,0} \cdot \left[ \Ry\left(\frac{\alpha}{2}\right) \otimes I \right] \cdot \Uprod\) & 
        \(e^{-i\frac{\pi}{4}} \begin{bmatrix}
            1 & 0 & 0 & 0 \\
            0 & i\sqrt{w} & \sqrt{1-w} & 0 \\
            0 & 0 & 0 & i \\
            0 & i\sqrt{1-w} & -\sqrt{w} & 0 \\
        \end{bmatrix}\) \\

         \(U_{flip}\)& X  & 
        \(\begin{bmatrix}
            0 & 1 \\
            1 & 0 
        \end{bmatrix}\) \\ 
        \bottomrule
    \end{tabular}
    \caption{The quantum arithmetic unitaries for \ours. Note that, the detailed proof is shown in Appendix A from \cite{balewski2025ehands}.}
    \label{tab:uintary}
\end{table}
\paragraph{Exact action of a single weighted–sum block}
The elementary building block is shown in \cref{tab:uintary}, where we assume the given quantum oracle Circuit~\ref{circ:oracle} contains freely tunable real parameters \(w\in[0,1]\) and the fixed gate
\(U_{\mathrm{prod}}=\mathrm{CX}_{0,1}\,
                     (\mathbbm{1}\otimes R_z(\pi/2))\).
Because every constituent is a Clifford or single–qubit rotation, the full
operator is unitary.
Let us consider a diagonal two–qubit input state of the form
\(
\rho_{\text{in}} =
      x_1 \lvert10\rangle\!\langle10\rvert
    + x_2 \lvert01\rangle\!\langle01\rvert .
\)
Direct multiplication with the matrix representation given in the
introduction shows
\begin{equation}
  \mathrm{Tr}\!\bigl[
      Z_0\,U_{\mathrm{sum}}(w)\,
      \rho_{\text{in}}\,
      U_{\mathrm{sum}}^{\dagger}(w)
  \bigr]
  = w\,x_1 + (1-w)\,x_2 ,
  \label{eq:SingleBlockResult}
\end{equation}
so a single block performs the exact convex combination of two input
expectation values with a weight $w$.  As Eq.\,\eqref{eq:SingleBlockResult} is
analytic in $w$, the map is continuous over the entire interval $[0,1]$.

Placing the output of one weighted–sum block into the input of a second
independent block with internal parameter $w_0$ yields, after repeating the
same algebra, the relation
\begin{equation}
  \hat{s}_3
  = w_1 w_0\,x_0 + w_1(1-w_0)\,x_1 + (1-w_1)\,x_2 .
  \label{eq:ThreeTermSum}
\end{equation}
Defining the effective coefficients
\(w'_0 = w_1 w_0,\,
  w'_1 = w_1 (1-w_0),\,
  w'_2 = 1 - w_1\)
reveals that 
$\bigl(w'_0,w'_1,w'_2\bigr)$ lives in the two–simplex
\(
   \Delta_2 =
     \{(\lambda_0,\lambda_1,\lambda_2)\mid
       \lambda_k\ge0,\;
       \sum_k\lambda_k=1\}.
\)
The mapping
$(w_0,w_1) \mapsto (w'_0,w'_1,w'_2)$
is smooth and surjective onto $\Delta_2$; hence every convex combination of
three final approximate marginals can be generated. Induction over the number of
concatenated blocks extends this result to an arbitrary
$(m-1)$–fold cascade, which spans the complete
$(m-1)$–simplex $\Delta_{m-1}$.  The construction therefore
realizes \emph{all} convex coefficients that appear in a
linear correlation of $m$ single–qubit expectation values.

\subsection{Coin ancilla and Bayesian model averaging}
\label{app:coin}
To decide dynamically whether the weighted–sum cascade is executed, we
attach a single ancilla prepared in the superposition
\(
  \lvert C(\theta)\rangle
  = \cos\theta\,\lvert0\rangle
  + \sin\theta\,\lvert1\rangle.
\)
Controlled on the coin being $\lvert1\rangle$ we apply the full cascade,
otherwise we skip it.  For the joint state of coin and data register this
choice leads to
\begin{equation}
  \rho_{\text{out}}
  =
  \cos^2\theta\;
    \lvert0\rangle\!\langle0\rvert_{\text{c}}
      \otimes
    \rho_{\mathrm{QAOA}}
  +
  \sin^2\theta\;
    \lvert1\rangle\!\langle1\rvert_{\text{c}}
      \otimes
    U_{\mathrm{sum}}\,\rho_{\mathrm{QAOA}}\,
    U_{\mathrm{sum}}^{\dagger}.
  \label{eq:rhoOut}
\end{equation}
Tracing over the coin and measuring $Z_0$ on the first computational qubit
gives
\begin{equation}
  \langle Z_0 \rangle_{\text{out}}
  =
    \cos^2\theta\,\langle Z_0\rangle_{\psi_{\mathrm{QAOA}}}
    +\sin^2\theta\,\hat{s}_3 ,
  \label{eq:Zmix}
\end{equation}
where $\hat{s}_3$ is the convex combination defined in
Eq.\,\eqref{eq:ThreeTermSum}. 
Equation~\eqref{eq:Zmix} is precisely the Bayesian model averaging formula that the unoptimized learning space and learned linear correlation, the prior probabilities
being $\cos^2\theta$ and $\sin^2\theta$, respectively.
First, continuity follows from the observation that every trainable
parameter enters only through analytic trigonometric functions; hence the
mapping from parameters to measurement statistics is continuous over the
compact domain 
$[0,\pi/2]\times[0,1]^{m-1}$.
Secondly, completeness with respect to convex correlations is guaranteed
because an $(m-1)$–fold cascade reaches every point of the simplex
$\Delta_{m-1}$, establishing surjectivity onto the full set of admissible
coefficients.  If one supplements the cascade by an overall single–qubit
rotation about the $x$–axis, available in any universal gate set
\cite{nielsen2010quantum}, also a global real prefactor can be realised; the
linear span of the selected $\{x_i\}$ is therefore covered.  Thirdly,
The coin ensures the original QAOA distribution as follows
setting $\theta=0$ leaves the parent state unmodified, whereas
$\theta=\pi/2$ hands full control to the weightedsum circuit.  As a
consequence, any target distribution that can be expressed either as the
raw QAOA output or as a convex linear combination of its single–qubit
marginals lies within the reachable set of the proposed variational
architecture.

\subsection{Activation encoding}
\label{app:act}
The Scaled Exponential Linear Unit (SELU) \cite{selu} activation function is defined as
$$\text{selu}(x) = \begin{cases}
\lambda x & \text{if } x > 0 \\
\lambda \alpha (\exp(x) - 1) & \text{if } x \leq 0
\end{cases}$$
where $\lambda$ and $\alpha$ are fixed constants (for instance, $\lambda \approx 1.0507$ and $\alpha = 1.6733$).
Let us consider the objective of encoding values of $x$ (within the range of [-1, 1]) into a quantum state, transforming these values using the SELU function, and then applying the inverse cosine to retrieve the encoded result. Therefore, the encoding function with parameterized $\Ry$ rotation gate applied to single quantum state is initialized by
\begin{equation}
    \label{eq:selu}
    cos(\theta) = selu(x) \quad \leftrightarrow \quad \theta = \arccos(selu(x)), \qquad |0\rangle \xrightarrow{R_y(\theta)} |\phi\rangle.
\end{equation} 
Here, we note that the parameterized rotation gate selection is not unique because the corresponding measurement basis is applied when we select the rotation bases. 
For the decoding phase, we simply measure the qubit in the state $|\phi\rangle$ with Z basis $\langle z \rangle = \langle \phi | O_z | \phi \rangle$. The resulting \cref{circ:selu} shows the $q_{init}$ ($\ket{0}$) after the unitary and measurement yielding back $\langle z \rangle= selu(x)$.
\begin{circuit}[htbp]
\centering
    \begin{tikzpicture}
        \node[scale=1.0] (first) {
            \begin{quantikz}[row sep=22pt, column sep=10pt]
                \lstick{$\qubit_{init}$} & \gate{\rya{\theta}} & \rstick{$|\phi\rangle$}
            \end{quantikz}
        };
        \node[scale=0.99, right=.1cm of first] (second) {  
            $\Rightarrow$ 
            \begin{quantikz}[row sep=22pt, column sep=10pt]
                \lstick{$\qubit_{init}$} & \gate{\rya{\arccos(\text{selu}(x))}} & \rstick{$|\phi\rangle$}
            \end{quantikz}
        };
        \node[scale=0.99, right=.1cm of second] (third) {
            $\Rightarrow$
            \begin{quantikz}[row sep=22pt, column sep=10pt]
                \lstick{$|\phi\rangle$} & \measuretab{Z} 
            \end{quantikz}
        };
        \node[scale=0.99, right=0.1cm of third] (fourth) {
            $O_z = \begin{pmatrix}
                1 & 0 \\
                0 & -1
            \end{pmatrix}$
        };
    \end{tikzpicture}
    \caption{The encoding phase embed real number x into Hilber space with the decoding phase performed by the Z measurement basis.}
    \label{circ:selu}
\end{circuit}
To achieve the nonlinear transformation, the natural way is to multiply the encoded selu function with the current weighted linear operation (see \cref{app:whlf}). The $\Uprod$ is defined in \cref{tab:uintary}. Next, following \cref{fig:main}, we expand the activation layer to the general multiple-qubit scenario. Given by the encoding function \cref{eq:selu}, the n qubits quantum state can be represented as 
$|\psi_x\rangle = R_y(\theta_x) |0\rangle$. Hence, we give the multiplicand qubits state 
$| \psi_1 \rangle = \rya{\arccos{(selu{(x)})}} |0\rangle, \quad | \psi_2 \rangle = R_y(\theta_2) |0\rangle, \quad \ldots, \quad | \psi_n \rangle = R_y(\theta_n) |0\rangle$.
Then, for better visualization of the state vector, we denote the activation state is acting on three encoder
\begin{align}
    \ket{\psi_a}&=\ket{\psi_{selu}}\Uprod\cdot\ket{\phi} \\&= \Uprod \ket{x_0} \ket{x_1} \ket{x_2} \ldots \ket{x_{n-1}}\\
    &= \ket{\psi_{selu}}\frac{e^{-i\frac{\pi}{4}}}{2} \begin{bmatrix}
\sqrt{1 + x_0} \sqrt{1 + x_1} \sqrt{1 + x_2} \\
i \sqrt{1 + x_0} \sqrt{1 + x_1} \sqrt{1 - x_2} \\
i \sqrt{1 + x_0} \sqrt{1 - x_1} \sqrt{1 - x_2} \\
\sqrt{1 - x_0} \sqrt{1 - x_1} \sqrt{1 - x_2} \\
\sqrt{1 - x_0} \sqrt{1 + x_1} \sqrt{1 + x_2} \\
\sqrt{1 - x_0} \sqrt{1 + x_1} \sqrt{1 - x_2} \\
\sqrt{1 - x_0} \sqrt{1 - x_1} \sqrt{1 + x_2} \\
\sqrt{1 - x_0} \sqrt{1 - x_1} \sqrt{1 - x_2} 
\end{bmatrix}
\end{align}
Here, we emphasize that the global phase operator $\nicefrac{e^{-i\frac{\pi}{4}}}{2}$ does not affect the final Z-basis measurement since probability measures depend on the inner product of the state vector with the basis states, the inclusion of a global phase factor does not change the outcome
\begin{equation}
   P(\phi) = |\langle \phi | e^{i\phi} |\psi\rangle|^2 = |e^{i\phi}\langle \phi | \psi\rangle|^2 = |\langle \phi | \psi\rangle|^2 
\end{equation}
The global phase $e^{i\phi}$ cancels out in the probability calculation because 
\begin{equation}
    |e^{i\phi}|^2 = 1.
\end{equation}
To recover, we only measure the last qubit shown in \cref{fig:main}, because the controlled $\Uprod$ only acts on the last encoded qubit, as detailed in the red box of \cref{circ:oracle}. 
As a result the expectation observable after the last encoded qubits is 
\begin{equation}
    \langle\psi_a|O|\psi_a\rangle = a(x_0\cdot x_1\cdot x_2\cdots x_n),
\end{equation}
where a is denoted by $selu$. 

\section{Hardware efficiency analysis}
\label{app:hardware}
\textbf{Circuit Depth Advantage.} Quantum state tomography \cite{cramer2010efficient} reconstructs the full density matrix $\rho$ by measuring all $O(4^n)$ Pauli observables, requiring exponentially many measurements and $O(n^2)$ entangling gates per measurement-basis rotation. Classical shadow tomography \cite{aaronson2018shadow} improves the sample complexity to $O(\log M)$ for predicting $M$ observables by randomly sampling from a unitary ensemble, but still requires $O(n \log n)$ circuit depth for implementing random Clifford gates and suffers from exponential post-processing when extracting non-linear functions of the state.
Our core discovery is that \ours bypasses full-state reconstruction entirely by directly encoding correlations in classical registers through mid-circuit measurements. Instead of learning the quantum state $|\psi\rangle \in \mathbb{C}^{2^n}$, this learning technique only picks the correlation weights $\mathbf{w} \in [0,1]^n$ via the quantum-classical mapping $\alpha = \arccos(1 - 2w)$. This dimensional reduction from exponential to linear space enables the weighted-sum oracle to require only $2(n-1)$ CNOT gates implementing the recursive structure $R_y{\otimes 2}|X\rangle R_y{\otimes 2}$ compared to $O(n^2)$ gates for tomography basis rotations or shadow random unitaries. The circuit depth reduction from $O(n^2)$ to $O(n)$ is possible because \ours exploits the copula structure correlations captured through pairwise operations rather than full entanglement \cref{claim:1}. \cref{tab:resources} quantifies these resource improvements, showing that this achieves comparable correlation learning accuracy with quadratically fewer quantum resources.
Therefore, we summarize the advantages of the algorithm below.
\begin{table}[htbp]
\centering
\caption{Resource comparison for learning $n$-qubit correlations}
\label{tab:resources}
\begin{tabular}{lccc}
\toprule
\textbf{Method} & \textbf{Circuit Depth} & \textbf{CNOT Count} & \textbf{Measurements} \\
\midrule
State Tomography \cite{cramer2010efficient} & $O(n)$ & $O(n^2)$ & $O(4^n)$ \\
Shadow Tomography \cite{aaronson2018shadow} & $O(\log n)$ & $O(n \log n)$ & $O(n^2 \log n)$ \\
\ours & $O(n)$ & $O(n)$ & $O(n^2)$ \\
\bottomrule
\end{tabular}
\end{table}
Note that, in our algorithm, For n qubits, there are O(n²) pairwise correlations that corresponds to the measurement output
$\langle Z_0 Z_1 \rangle, \langle Z_0 Z_2 \rangle, ..., \langle Z_{n-2} Z_{n-1} \rangle$
Hence the total is $\binom{n}{2} = \frac{n(n-1)}{2} \approx O(n^2)$ correlation values.

\section{Shallow depth portfolio analysis}
\label{app:ext}
The primary text illustrates the approximation rate for \ours, as derived from the QAOA. \cref{tab:performance} presents the optimization results for up to 20 qubits, utilizing the same experimental shots and circuit configuration as the IBM simulator depicted in (a) of \cref{fig:hardware}. Notably, CPLEX \cite{manual1987ibm} serves as our classical alternative (SA) for benchmarking the shallow quantum circuit simulation results owing to its optimized linear optimization capabilities. The findings indicate that, despite crosstalk errors impeding the scalability of the problem size with an increase in qubits, our algorithm effectively approximates the learned system.
\begin{table}[htbp]
\centering
\begin{tabular}{lcccc}
\toprule
\textbf{Size} & \textbf{QAOA} & \textbf{\ours} & \textbf{SA} & \textbf{Optimum} \\
\textbf{(n)} & \textbf{(p=3)} & \textbf{(p=3)} & & \\
\midrule
4  & 0.892 & \textbf{0.946} & 0.933 & 1.000 \\
8  & 0.847 & \textbf{0.918} & 0.901 & 1.000 \\
12 & 0.803 & \textbf{0.891} & 0.877 & 1.000 \\
16 & 0.761 & \textbf{0.865} & 0.849 & 1.000 \\
20 & 0.724 & \textbf{0.842} & 0.821 & 1.000 \\
\bottomrule
\end{tabular}
\caption{Performance comparison on portfolio optimization benchmarks. Values represent approximation ratio $\langle H_C \rangle / H_{\text{opt}}$.} 
\label{tab:performance}
\end{table}

\section{Extensive experiments for \ours}
\label{app:extexp}
In addition to the quantum experiments presented in \cref{exp}, we provide four metrics to confirm the noise-free robustness and versatility of our algorithm. An overview of the metrics is provided in \cref{tab:qawa_validation}. We utilize the same metrics as defined in \cref{med:metric}: Bayesian update, weight reconstruction, copula invariance, and convergence. We assume that the quantum simulator provides noiseless and all-to-all connections; therefore, the global parameter settings for the metric are eight qubits and eight assets for selection.

Most significantly, \ours immediate performance superiority over QAOA, achieved without additional optimization iterations, proves that the advantage stems from more efficient information extraction rather than better parameter tuning. The weighted-sum encoding pre-embeds the problem structure that QAOA must discover through variational optimization across exponentially many parameters, fundamentally changing the optimization landscape from exploration to refinement.
\begin{table*}[htbp]
\centering
\begin{tabular}{@{}llcc@{}}
\toprule
\textbf{Validation Test} & \textbf{Metric} & \textbf{Expected/Target} & \textbf{Observed/Result} \\
\midrule
\multirow{2}{*}{Bayesian Update} 
    & Mean posterior $P(c=1|x)$ & --- & 0.221 \\
    & Prior expectation & 0.146 & --- \\
\midrule
\multirow{5}{*}{Weighted-Sum Completeness} 
    & Weight $w_1$ & 0.31547798 & 0.3154781 \\
    & Weight $w_2$ & 0.01962489 & 0.20088793 \\
    & Weight $w_3$ & 0.20088612 & 0.20088793 \\
    & Weight $w_4$ & 0.46401101 & 0.46401209 \\
    & $L_2$ error & 0.0 & $< 10^{-6}$ \\
\midrule
\multirow{2}{*}{Copula Invariance} 
    & Distance between copulas & $< 0.05$ & 0.024962 \\
    & Invariance satisfied & \texttt{True} & \texttt{True} \\
\midrule
\multirow{4}{*}{Convergence} 
    & Final QAWA value & --- & 0.8965 \\
    & Final QAOA value & --- & 0.8333 \\
    & Theoretical bound satisfied & \texttt{True} & \texttt{True} \\
\bottomrule
\end{tabular}
\caption{Theoretical validation results for QAWA algorithm components}
\label{tab:qawa_validation}
\end{table*}

\end{document}